\documentclass[11pt]{article}
 \pdfoutput=1

\usepackage{fullpage}
\usepackage{amscd}
\usepackage{amsmath}
\usepackage{amssymb}
\usepackage{xspace}
\usepackage{theorem}
\usepackage{graphicx}
\usepackage{ifpdf}
\usepackage{url,hyperref}
\usepackage{latexsym}
\usepackage{euscript}
\usepackage{xspace}
\usepackage[all]{xy}
\usepackage{color}
\usepackage{makeidx}
\usepackage{picins,wrapfig}
\usepackage[utf8]{inputenc}

\newcommand{\remove}[1]{}

\usepackage[utf8]{inputenc}

\usepackage{relsize,lineno}
\usepackage[ruled,vlined]{algorithm2e}
\usepackage{mathrsfs,dsfont}
\usepackage{epsfig}

\newtheorem{theo}{Theorem}

\newtheorem{claim}[theo]{Claim}

\newtheorem{lemma}[theo]{Lemma}
\newtheorem{defi}[theo]{Definition}
\newtheorem{coro}[theo]{Corollary}

\newenvironment{proof}[1][{}]{
  \begin{trivlist}\item[]\textit{Proof #1}\quad}%
  {\hfill\hspace*{\fill}~$\square$\end{trivlist}}

\begin{document}

\title{
Linear kernels for $k$-tuple and liar's domination\\
in bounded genus graphs}

\author{
Arijit Bishnu\\
{\normalsize ACM Unit}\\
{\normalsize Indian Statistical Institute}\\
{\normalsize Kolkata, INDIA}\\
\url{arijit@isical.ac.in}
\and
Arijit Ghosh
\footnote{
Supported by the Indo-German Max Planck Center for Computer Science (IMPECS).
}
\thanks{
Part of this
work was done when the author was a visiting scientist in ACM Unit,
Indian Statistical Institute, Kolkata.}\\
{\normalsize D1: Algorithms \& Complexity}\\
{\normalsize Max-Planck-Institut f\"ur Informatik}\\
{\normalsize Saarbr\"ucken, Germany}\\
\url{agosh@mpi-inf.mpg.de}
\and
Subhabrata Paul\\
{\normalsize ACM Unit}\\
{\normalsize Indian Statistical Institute}\\
{\normalsize Kolkata, INDIA}\\
\url{paulsubhabrata@gmail.com}
}
\maketitle

\pagenumbering{roman}

\begin{abstract}
A set $D\subseteq V$ is called a \emph{$k$-tuple dominating set} of a graph $G=(V,E)$ if 
$\left| N_G[v] \cap D \right| \geq k$ for all $v \in V$, where
$N_G[v]$ denotes the closed neighborhood of $v$. A set $ D \subseteq V$ is called a \emph{liar's dominating set} 
of a graph $G=(V,E)$ if 
\begin{itemize}

\item[]~\,(i)
$\left| N_G[v] \cap D \right| \geq 2$ for all $v\in V$, and 

\item[]~(ii)
for every pair of 
distinct vertices $u, v\in V$, $\left| (N_G[u] \cup N_G[v]) \cap D \right| \geq 3$. 

\end{itemize}
Given a graph $G$, the decision 
versions of \textsc{$k$-Tuple Domination Problem} and the \textsc{Liar's Domination Problem} are to check whether 
there exists a $k$-tuple dominating set and a liar's dominating set of $G$ of a given cardinality, respectively. These 
two problems are known to be \textsf{NP}-complete~\cite{LiaoChang2003, Slater2009}. In this paper, we study the parameterized 
complexity of these problems. We show that the \textsc{$k$-Tuple Domination Problem} and the 
\textsc{Liar's Domination Problem} are \textsf{W[2]}-hard for general graphs but they admit linear kernels for 
graphs with bounded genus.
\end{abstract}


\clearpage

\pagenumbering{arabic}

\section{Introduction}
Let $G=(V,E)$ be a graph. For a vertex $v\in V$, let $N_G(v)=\{u\in V
|uv \in E\}$ and $N_G[v] = N_G(v) \cup \{v\}$ denote the
open and closed neighborhoods of $v$, respectively. A set $D\subseteq V$ is
called a \emph{dominating set} of a graph $G=(V,E)$ if $|N_G[v] \cap D| \geq 1$
for all $v \in V$. The \emph{domination number} of a graph $G$, denoted by
$\gamma(G)$, is the minimum cardinality of a dominating set of $G$. The concept
of domination has been well studied. Depending upon various applications,
different variations of domination have appeared in the literature
\cite{HaynesHedetniemiSlater1998,HaynesHedetniemiSlater11998}.

Among different variations of domination, $k$-tuple domination and liar's
domination are two important and well studied type of domination
\cite{Harary2000,LiaoChang2002,LiaoChang2003,RodenSlater2009,Slater2009}. A set
$D\subseteq V$ is called a \emph{$k$-tuple dominating set} of a graph $G=(V,E)$
if each vertex $v \in V$ is dominated by at least $k$ number of vertices in $D$,
that is, $|N_G[v] \cap D| \geq k$ for all $v \in V$. The concept of $k$-tuple
domination in graphs was introduced in~\cite{Harary2000}. For $k=2$ and $3$, it
is called \emph{double domination} and \emph{triple domination} respectively.
The \emph{$k$-tuple domination number} of a graph $G$, denoted by
$\gamma_{k}(G)$, is the minimum cardinality of a $k$-tuple dominating set of
$G$. It is a simple observation that for the existence of a $k$-tuple dominating
set, we need $\delta(G)\geq k-1$, where $\delta(G)$ is the minimum degree of
$G$. On the other hand, liar's domination is a new variation of domination and was introduced in 2009 by 
Slater~\cite{Slater2009}. A set $D\subseteq V$ is called a \emph{liar's dominating
set} of a graph $G=(V,E)$ if the following two conditions are met:
\begin{description}
 \item[condition (i)] $|N_G[v] \cap D| \geq 2$ for all $v\in V$
 \item[condition (ii)] for every pair of distinct vertices $u, v\in V$,
$|(N_G[u]\cup N_G[v])\cap D|\geq 3$
\end{description}
In a network guarding scenario, if sentinels are placed
in the vertices of the dominating set, then the graph (network) is guarded.
Consider the situation where a single sentinel is unreliable or lies and
we do not know the exact sentinel that lies. We then need a liar's dominating
set to guard the network. The \emph{liar's domination number} of a graph $G$,
denoted by $\gamma_{LR}(G)$, is the minimum cardinality of a liar's dominating
set of $G$. Formally, the decision versions of \textsc{$k$-Tuple Domination
Problem} and \textsc{Liar's Domination Problem} are defined as follows.

\underline{\textsc{$k$-Tuple Domination Problem}}
\begin{description}
\item{\emph{Instance:}} A graph $G=(V,E)$ and a nonnegative integer $p$.

\item{\emph{Question:}} Does there exist a $k$-tuple dominating set of cardinality at most $p$?
\end{description}

\underline{\textsc{Liar's Domination Problem}}
\begin{description}
\item{\emph{Instance:}} A graph $G=(V,E)$ and a nonnegative integer $p$.
\item{\emph{Question:}} Does there exist a liar's dominating set of cardinality
at most $p$?
\end{description}

Note that, every liar's dominating set is a double dominating set and every
triple dominating set is a liar's dominating set. Hence, liar's domination
number lies between double and triple domination number, that is,
$\gamma_{2}(G)\leq \gamma_{LR}(G)\leq \gamma_{3}(G)$.

The rest of the paper is organized as follows. Section $2$ introduces some
pertinent definitions and preliminary results that are used in the rest of the
paper and a brief review on the progress in the study of parametrization for
domination problems. Section $3$ deals with the hardness results of both
$k$-tuple domination problem and liar's domination problem. In Section $4$, we
show that both \textsc{$k$-Tuple Domination Problem} and \textsc{Liar's Domination
Problem} admit linear kernel in planar graphs. In Section $5$, we
extend the results for bounded genus graphs. Finally, Section $6$ concludes the
paper.

\section{Preliminaries}

Let $G=(V,E)$ be a graph. Let $G[S]$, $S \subseteq V$ denote the induced subgraph of $G$ on the vertex set $S$. The 
\emph{distance} between two vertices $u$ and $v$ in a graph $G$ is the number of edges in a shortest path connecting them and 
is denoted as $d_G(u,v)$. The degree of a vertex $v \in V(G)$, denoted by $deg_G(v)$, is the number of neighbors of $v$.

\subsection{Graphs on surfaces}
\label{subsec:graphsurfacenotation}

In this subsection, we recall some basic facts about graphs on surfaces following the discussion in~\cite{Fominkernelgenus}. 
The readers are referred to~\cite{Mohar2001} for more details. A \emph{surface} $\Sigma$ is a compact $2$-manifold without 
boundary. Let $\Sigma_0$ denote the sphere $\{(x,y,z)| ~x^2+y^2+z^2=1\}$. A \emph{line} and \emph{O-arc} are subsets of $\Sigma$ 
that are homeomorphic to $[0,1]$ and a circle respectively. A subset of $\Sigma$ meeting the drawing only in vertices of $G$ is 
called \emph{$G$-normal}. If an O-arc is G-normal, then it is called a \emph{noose}. The length of a noose is the number of its 
vertices. The \emph{representativity} of $G$ embedded in $\Sigma\neq \Sigma_0$ is the smallest length of a non-contractible 
noose in $\Sigma$ and it is denoted by \emph{rep$(G)$}.

The classification theorem for surfaces states that, any surface $\Sigma$ is homeomorphic to either a surface $\Sigma^h$ which 
is obtained from a sphere by adding $h$ handles (orientable surface), or a surface $\Sigma^k$ which is obtained from a sphere 
by adding $k$ crosscaps (non-orientable surface)~\cite{Mohar2001}. The \emph{Euler genus} of a non-orientable surface $\Sigma$, 
denoted by \emph{eg$(\Sigma)$}, is the number of crosscaps $k$ such that $\Sigma \cong \Sigma^k$ and for an orientable surface, 
\emph{eg$(\Sigma)$} is twice the number of handles $h$ such that $\Sigma \cong \Sigma^h$. Given a graph $G$, Euler genus of $G$, 
denoted by eg$(G)$, is the minimum eg$(\Sigma)$, where $\Sigma$ is a surface in which $G$ can be embedded. The \emph{Euler characteristic} 
of a surface $\Sigma$ is defined as $\chi(\Sigma)=2-\textup{eg}(\Sigma)$. For a graph $G$, $\chi(G)$ denotes the largest number $t$ 
for which $G$ can be embedded on a surface $\Sigma$ with $\chi(\Sigma)=t$. Let $G=(V,E)$ be a $2$-cell embedded graph in $\Sigma$, 
that is, all the faces of $G$ is homeomorphic to an open disk. If $F$ is the set of all faces, then \emph{Euler's formula} tells 
that $V-E+F=\chi(\Sigma)=2-\textup{eg}(\Sigma)$.

Next we define a process called \emph{cutting along a noose $N$}. Although the formal defi is given in~\cite{Mohar2001}, we 
follow a more intuitive defi given in~\cite{Fominkernelgenus}. Let $N$ be a noose in a $\Sigma$-embedded graph $G=(V,E)$. 
Suppose for any $v\in N\cap V$, there exists an open disk $\Delta$ such that $\Delta$ contains $v$ and for every edge $e$ 
adjacent to $v$, $e\cap \Delta$ is connected. We also assume that $\Delta-N$ has two connected components $\Delta_1$ and 
$\Delta_2$. Thus we can define partition of $N_G(v)= N_G^1(v)\cup N_G^2(v)$, where $N_G^1(v)= \{u\in N_G(v)| uv\cap \Delta_1\neq \emptyset\}$ 
and $N_G^2(v)= \{u\in N_G(v)| uv\cap \Delta_2\neq \emptyset\}$. Now for each $v\in N\cap V$ we do the following:
\begin{enumerate}
\item remove $v$ and its incident edges
\item introduce two new vertices $v^1, v^2$ and
\item connect $v^i$ with the vertices in $N_G^i$, $i = 1, 2$.
\end{enumerate}
The resulting graph $\mathcal{G}$ is obtained from $\Sigma$-embedded graph $G$
by cutting along $N$. The following lemma is very useful in the proofs by
induction on the genus.
\begin{lemma}\textup{\cite{Fominkernelgenus}} \label{lemcutting}
Let $G$ be a $\Sigma$-embedded graph and let $\mathcal{G}$ be a graph obtained
from $G$ by cutting along a non-contractible noose $N$. Then one of the
following holds
\begin{itemize}
\item $\mathcal{G}$ is the disjoint union of graphs $G_1$ and $G_2$ that can be
embedded in surfaces $\Sigma_1$ and $\Sigma_2$ such that $eg(\Sigma) =
eg(\Sigma_1) + eg(\Sigma_2)$ and $eg(\Sigma_i) > 0, i = 1, 2$.
\item $\mathcal{G}$ can be embedded in a surface with Euler genus strictly
smaller than eg$(\Sigma)$.
\end{itemize}
\end{lemma}

A \emph{planar graph} $G=(V,E)$ is a graph that can be embedded in the plane. We term such an embedding as a \emph{plane graph}.

\subsection{Parameterization and domination}
A \emph{parameterized problem} is a language $L\subseteq \Sigma^* \times \mathds{N}$, where $\Sigma^*$ denotes the set of all finite strings over a finite alphabet $\Sigma$. A parameterized problem $L$ is \emph{fixed-parameter tractable} if the question ``$(x,p)\in L$'' can be decided in time $f(p)\cdot |x|^{O(1)}$, where $f$ is a computable function on nonnegative integers, $x$ is the instance of the problem and $p$ is the parameter. The corresponding complexity class is called \textsf{FPT}. Next we define a reducibility concept between two parameterized problems.

\begin{defi} \textup{\cite{DowneyFellows,Niedermeier}}
Let $L, L'\subseteq \Sigma^*\times \mathds{N}$ be two parameterized problems. We say that $L$ reduces to $L'$ by a \emph{standard parameterized m-reduction} if there are functions $p\mapsto p'$ and $p\mapsto p''$ from $\mathds{N}$ to $\mathds{N}$ and a function $(x,p)\mapsto x'$ from $\Sigma^*\times \mathds{N}$ to $\Sigma^*$ such that
\begin{enumerate}
\item $(x,p)\mapsto x'$ is computable in time $p''|x|^c$ for some constant $c$ and
\item $(x,p)\in L$ if and only if $(x',p')\in L'$.
\end{enumerate}
\end{defi}

A parameterized problem is in the class \textsf{W[i]}, if every instance $(x, p)$ can be transformed (in fpt-time) to a combinatorial circuit that has height at most $i$, such that $(x, p)\in L$ if and only if there is a satisfying assignment to the inputs, which assigns $1$ to at most $p$ inputs. A problem $L$ is said to be \textsf{W[i]}\emph{-hard} if there exists a standard parameterized m-reduction from all the problems in \textsf{W[i]} to $L$ and in addition, if the problem is in \textsf{W[i]}, then it is called \textsf{W[i]}\emph{-complete}.

\remove{
The height is the largest number of logical units with unbounded fan-in on any
path from an input to the output. The number of logical units with bounded
fan-in on the paths must be limited by a constant that holds for all instances
of the problem.}

Next we define the reduction to problem kernel, also simply referred to as \emph{kernelization}.

\begin{defi} \textup{\cite{Niedermeier}}
Let $L$ be a parameterized problem. By reduction to problem kernel, we mean to replace instance $I$ and the parameter $p$ of $L$ by a ``reduced'' instance $I'$ and by another parameter $p'$ in polynomial time such that
\begin{itemize}
\item $p'\leq c\cdot p$,  where $c$ is a constant,
\item $I'\leq g(p)$, where $g$ is a function that depends only on $p$, and
\item $(I,p)\in L$ if and only if $(I',p')\in L$.
\end{itemize}
The reduced instance $I'$ is called the \emph{problem kernel} and the size of the problem kernel is said to be bounded by $g(p)$.
\end{defi}

In parameterized complexity, domination and its variations are well studied problems. The decision version of domination problem is \textup{W[2]}-complete for general graphs~\cite{DowneyFellows}. But this problem is FPT when restricted to planar graphs~\cite{Alberkernel} though it is still \textup{NP}-complete for this graph class~\cite{GareyJohnson}. Furthermore, for bounded genus graphs, which is a super class of planar graphs, domination problem remains FPT~\cite{Fominkernelgenus}. It was proved that dominating set problem possesses a linear kernel in planar graphs~\cite{Alberkernel} and in bounded genus graphs~\cite{Fominkernelgenus}. Also domination problem admits polynomial kernel on graphs excluding a fixed graph $H$ as a minor~\cite{Gutner2009} and on $d$-degenerated graphs~\cite{kerneldomdegenerate}.
A search tree based algorithm for domination problem on planar graphs, which runs in $O(8^p n)$ time, is proposed in~\cite{Albersearchtree}. For bounded genus graphs, similar search tree based algorithm is proposed in~\cite{Ellis2004} and has a time complexity of $O((4g+40)^p n^2)$, where $g$ is the genus of the graph. Algorithms with running time of $O(c^{\sqrt{p}}n)$ for domination problem on planar graphs have been devised in~\cite{Alber2002,Alber2004,Fomin2003,Fominkernelgenus}. Like domination problem, \textsc{$k$-Tuple Domination Problem} and \textsc{Liar's Domination Problem} are both \textsf{NP}-complete~\cite{LiaoChang2003,Slater2009} for general graphs. However, these problems have been polynomially solved for different graph classes~\cite{LiaoChang2002, LiaoChang2003, PandaPaultree, PandaPaulpig}. But for planar graphs and hence for graphs with bounded genus, \textsc{$k$-Tuple Domination Problem} remains \textsf{NP}-complete~\cite{Lee2008}. In~\cite{Slater2009}, though the \textsf{NP}-completeness proof is given for general graphs, it can be verified that using the same construction one can find the \textsf{NP}-completeness of 
\textsc{Liar's Domination Problem} in planar graphs, see Lemma~\ref{lem-liar-dom-planar-NP} in Appendix~\ref{sec-appendix}.

Some generalization of classical domination problem have been studied in the literature from parameterized point of view. Among those problems, $k$-dominating threshold set problem, $[\sigma, \rho]$-dominating set problem (also known as generalized domination) are generalized version of $k$-tuple dominating set problem. In~\cite{Golovachthresholddom2008}, it is proved that $k$-dominating threshold set problem is FPT in $d$-degenerated graphs. $[\sigma, \rho]$-domination is studied in~\cite{sigmarhohardness,Telle1997,Rooijsigmarhodom2009}. A set $D$ of vertices of a graph $G$ is $[\sigma,\rho]$-dominating set if for any $v \in D, |N(v)\cap D| \in \sigma$ and for any $v\notin D, |N(v)\cap D| \in \rho$ for any two sets $\sigma$ and $\rho$. It is known that $[\sigma, \rho]$-domination is FPT when parameterized by treewidth~\cite{Rooijsigmarhodom2009}. By Theorem 32 of~\cite{Alber2002}, it follows that $k$-tuple domination is FPT on planar graphs. But there is no explicit kernel for both $k$-tuple domination and liar's domination problem in the literature.

There have been successful efforts in developing meta-theorems like the celebrated Courcelle's theorem~\cite{Courcelle92} which states that all graph properties definable in monadic second order logic can be decided in linear time on graphs of bounded tree-width. This also implies FPT algorithms for bounded tree-width graph for these problems. In case of kernelization in bounded genus graphs, Bodlaender et al. give two meta-theorems~\cite{Bodlaender2009}. The first theorem says that all problems expressible in counting monadic second order (CMSO) logic and satisfying a coverability property admit a polynomial kernel on graphs of bounded genus and the second theorem says that all problems that have a finite integer index and satisfy a weaker coverability property admit a linear kernel on graphs of bounded genus. It is easy to see that both $k$-tuple and liar's domination problems can be expressed in CMSO logic. Let $G=(V,E)$ be an instance of a graph problem $\Pi$ such that $G$ is embeddable in a surface of Euler genus at most $r$. The basic idea of quasi-coverable property for $\Pi$ is that there exists a set $S\subseteq V$ satisfying the conditions of $\Pi$ such that the tree-width of $G\setminus R^r_G(S)$ is at most $r$ where $R^r_G(S)$ is a special type of reachability set from $S$. In domination type of problems, this reachability set is actually the whole graph and hence these problems satisfy the quasi-coverable property. The basic idea of strong monotonicity for a graph problem $\Pi$ is roughly as follows: Let $\mathcal{F}_i$ be a class of graphs $G$ having a specific set of vertices $S$ termed as the boundary of $G$ such that $|S|=i$. The glued graph $G=G_1\oplus G_2$ of $G_1$ and $G_2$ is the graph which is obtained by taking the disjoint union of $G_1$ and $G_2$ and joining $i$ edges between the vertices of the boundary sets. A problem $\Pi$ is said to satisfy the strong monotonicity if for every boundaried graph $G=(V,E)\in \mathcal{F}_i$, there exists a set $W\subseteq V$ of a specific cardinality which satisfy the property of $\Pi$ such that for every boundaried graph $G'=(V',E')\in \mathcal{F}_i$ with a set $W' \subseteq V'$, satisfying the property of $\Pi$, the vertex set $W\cup W'$ satisfies the property of $\Pi$ for the glued graph $G=G \oplus G'$. It can be verified easily that both $k$-tuple domination and liar's domination problems satisfy the strongly monotone property. The strongly monotone property implies the finite integer index for these problems. Hence, by the second meta-theorem in~\cite{Bodlaender2009}, both $k$-tuple and liar's domination problems admit linear kernels for graphs on bounded genus. Though these meta-theorems provide simple criteria to decide whether a problem admits a linear or polynomial kernel, finding a linear kernel with reasonably small constants for a specific problem is a worthy topic of further research~\cite{Bodlaender2009}. In this paper, we have obtained linear kernels with small constants for both the problems on bounded genus graphs. We have also proved the \textsf{W[2]}-hardness for $k$-tuple and liar's domination for general graphs.

\section{Hardness results in general graphs}

In this section, we show that \textsc{$k$-tuple Domination Problem} and \textsc{Liar's Domination Problem} 
are \textsf{W[2]}-hard. In~\cite{sigmarhohardness}, it is proved that $[\sigma,\rho]$-domination problem 
for any recursive sets $\sigma$ and $\rho$ is \textsf{W[2]}-hard. This implies the hardness for $k$-tuple 
domination in general graphs. But in this paper, we have come up with a simple \textsf{W[2]}-hardness proof 
for $k$-tuple domination in general graphs. To prove this, we show standard parameterized m-reductions from 
\textsc{Domination Problem}, which is known to be \textsf{W[2]}-complete~\cite{DowneyFellows}, to 
\textsc{$k$-Tuple Domination Problem} and \textsc{Liar's Domination Problem}, respectively.

\begin{theo}
\textsc{$k$-Tuple Domination Problem} is $\mathsf{W[2]}$-hard.
\end{theo}
\begin{proof}
We show a standard parameterized m-reduction from \textsc{Domination Problem} to \textsc{$k$-Tuple
Domination Problem}. Let $<G=(V,E),p>$ be an instance of \textsc{Domination
Problem}. We construct an instance $<G'=(V',E'), p'>$ of the \textsc{$k$-Tuple
Domination Problem} as follows: $V' = V \cup V_k$ where $V_k = \{u_1, u_2,
\ldots, u_k\}$ and $ E' = E\cup \{v_iu_j| v_i \in V \mbox{ and } u_j \in V_k
\setminus u_k \} \cup \{u_iu_j| u_i, u_j \in V_k, i \not= j\}$. Also set
$p'=p+k$. The construction of $G'$ from $G$ in case of triple domination is
illustrated in Figure \ref{figk-dom}.

\begin{figure}[h]
  \centering
  \includegraphics[width=.5\textwidth]{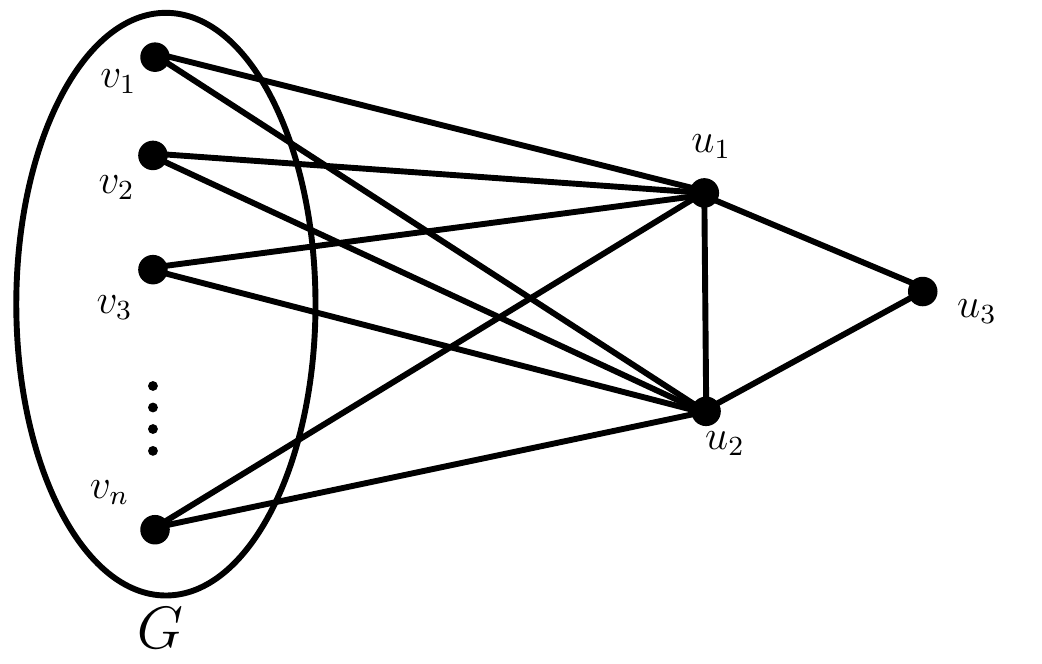}
	\caption{Construction of $G'$ from $G$ for triple domination}
	\label{figk-dom}
\end{figure}

\begin{claim} 
$G$ has a dominating set of size at most $p$ if and only if
$G'$ has a $k$-tuple dominating set of size at most $p'$.
\end{claim}

\begin{proof} 
Let $D$ be a dominating set of $G$ of cardinality
at most $p$ and $D'= D \cup V_k$.
Each $v_i \in V$ is dominated by at least one vertex from $D$ and by $k-1$
vertices from $V_k$. Each $u_i \in V_k$ is dominated by $k$ vertices of $V_k$.
Thus, $D'$ is a $k$-tuple dominating set of $G'$ of cardinality at most
$p'$.
\remove{It is easy to verify
that for each vertex $x\in V'$, $|N_{G'}[x]\cap D'|\geq k$, i.e., $D'$ is a
$k$-tuple dominating set of $G'$ of cardinality $p+k=p'$.
}

Conversely, let $D'$ be a $k$-tuple dominating set of $G'$ of cardinality at most $p'$.
Note that each $k$-tuple dominating set contains the set $V_k$ because to
dominate $u_k$ by $k$ vertices we must select all the vertices of $V_k$. Let
$D = D'\setminus V_k$. Clearly $D \subseteq V$ and $|D| \leq p$. Now for each
$v\in V$, $|N_G[v]\cap D|\geq 1$ because otherwise, there exists a vertex $v\in
V$ such that $|N_{G'}[v]\cap D'|=k-1$. This is a contradiction because $D'$
is a $k$-tuple dominating set of $G'$.  Thus $D$ is a dominating set of $G$ of
cardinality at most $p$.

Hence, $G$ has a dominating set of size at most $p$ if and only if $G'$ has a $k$-tuple
dominating set of size at most $p'$.
\end{proof}

Thus, \textsc{$k$-Tuple Domination Problem} is \textsf{W[2]}-hard.
\end{proof}

Next we show the \textsf{W[2]}-hardness of \textsc{Liar's Domination Problem}.

\begin{theo}
\textsc{Liar's Domination Problem} is $\mathsf{W[2]}$-hard.
\end{theo}
\begin{proof}
We show a standard parameterized m-reduction from \textsc{Domination Problem} to \textsc{Liar's
Domination Problem}. Let $<G=(V,E),p>$ be an instance of \textsc{Domination
Problem}. We construct an instance $<G'=(V',E'), p'>$ of the \textsc{Liar's
Domination Problem} as follows: $V'=V\cup \{u, u', v, v', w\}$ and $E'=E\cup
\{v_iu|v_i \in V\}\cup \{v_iv|v_i \in V \} \cup \{uu', vv', wu, wv\}$.
Also $p'=p+4$. The construction of $G'$ from $G$ is illustrated in Figure
\ref{fig liar}.

\begin{figure}[h]
	\centering
  \includegraphics[width=.5\textwidth]{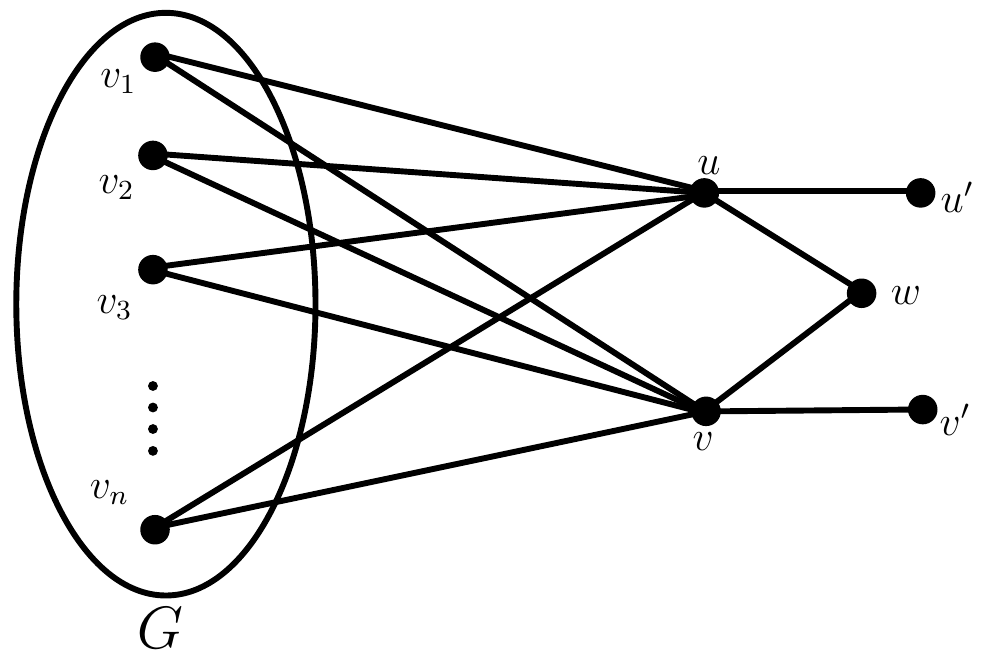}
	\caption{Construction of $G'$ from $G$}
	\label{fig liar}
\end{figure}

\begin{claim} 
$G$ has a dominating set of size at most $p$ if and only if
$G'$ has a liar's dominating set of size at most $p'$.
\end{claim}

\begin{proof}
Let $D$ be a dominating set of $G$ of cardinality
at most $p$ and $D'=D\cup \{u, u', v, v'\}$. It is easy to verify that for each
vertex $x\in V'$, $|N_{G'}[x]\cap D'|\geq 2$ and for every pair of vertices
$x,y\in V'$, $|(N_{G'}[x]\cup N_{G'}[y])\cap D'|\geq 3$. Hence $D'$ is a liar's
dominating set of $G'$ of cardinality at most $p+4=p'$.

Conversely, let $D'$ be a liar's dominating set of $G'$ of cardinality at most
$p'$. Each liar's dominating set must contain the set $\{u, u', v, v'\}$
because to doubly dominate $u'$ and $v'$ we must select the vertices $\{u, u',
v, v'\}$.

Let $X \subseteq V$ denote the set of vertices that are dominated by
exactly two vertices ($u$ and $v$) from $D'$. We claim $|X| \le 1$.
If there exists two such vertices $x,y \in V$, then
$|(N_{G'}[x] \cup N_{G'}[y]) \cap D'| = 2$ which violates condition (ii) of
liar's domination. We now deal with two cases:
\begin{description}
 \item[$|X| = 1 $:] Let $X=\{x\}$. Here $|N_{G'}[x]\cap D'|=2$. This implies
$w \in D'$, otherwise the pair $x$ and $w$ violates condition (ii) of liar's
domination. We set $D''=(D' \setminus \{w\})\cup \{x\}$. $D''$ is also a liar's
dominating set of $G'$ of cardinality at most $p'$. Note that all vertices in
$V$ is triply dominated by $D''$ and it does not contain $w$.

 \item[$|X| = 0 $:] In this case each vertex of $V$ is triply dominated by $D'$. Now if $w\notin D'$, we are done. Otherwise 
the set $D'\setminus \{w\}$ forms a liar's dominating set of $G'$ of cardinality at most $p'$ such that each vertex of $V$ 
is triply dominated by $D'\setminus \{w\}$.
\end{description}
Hence without loss of generality, we assume that there is a liar's dominating set $D'$
of $G'$ of cardinality at most $p'$ such that every vertex in $V$ is triply
dominated by $D'$ and $w\notin D'$. Let $D=D'\setminus \{u, u', v, v'\}$.
Clearly $D \subseteq V$ and $|D| \leq p$. Now for each $x\in V$, $|N_G[x]\cap D|\geq 1$ because otherwise, there exists a 
vertex $x\in V$ such that for the pair $x$ and $w$, condition (ii) of liar's domination is violated. This is a contradiction 
because $D'$ is a liar's dominating set of $G'$.  Thus $D$ is a dominating set of $G$ of cardinality at most $p$.

Hence, $G$ has a dominating set of size at most $p$ if and only if $G'$ has a liar's dominating set of size at most $p'$. 
\end{proof}

Thus, \textsc{Liar's Domination Problem} is \textsf{W[2]}-hard.
\end{proof}

\section{Linear kernels for planar graphs}

Having seen that $k$-tuple and liar's domination are \textsf{W[2]}-hard
in general graphs, we focus on planar graphs in this section and show
that they are FPT.

\subsection{Double domination}

In this subsection we show that \textsc{Double Domination Problem} in planar
graphs possesses a linear kernel. Our proof technique uses the region
decomposition idea of Alber et al.~\cite{Alberkernel}. First we describe the
reduction rules for kernelization.

\subsubsection{Reduction rule}

Let $G=(V,E)$ be the instance for \textsc{Double Domination Problem}. Consider
a pair of vertices $u,v\in V$. Let $N_G(u,v)=N_G(u)\cap N_G(v)$. We partition
the vertices of $N_G(u,v)$ in to three parts as follows:
\begin{eqnarray*}
N^1_G(u,v)&=&\{x\in N_G(u,v)| N_G(x)\setminus \{N_G(u,v) \cup \{u,v\} \}\neq
\emptyset\};\\
N^2_G(u,v)&=&\{x\in N_G(u,v)\setminus N^1_G(u,v) | N_G(x)\cap N^1_G(u,v) \neq \emptyset\};\\
N^3_G(u,v)&=& N_G(u,v)\setminus (N^1_G(u,v)\cup N^2_G(u,v)).
\end{eqnarray*}

\noindent{\textbf{Reduction Rule:}}
For every pair of distinct vertices $u,v\in V$, if $N^3_G(u,v) \neq \emptyset$, then
\begin{itemize}
\item delete all the vertices of $N^2_G(u,v)$ and
\item delete all vertices of $N^3_G(u,v)$ except one vertex.
\end{itemize}

\begin{lemma}
Let $G=(V,E)$ be a graph and $G'=(V', E')$ be the resulting graph after having
applied the reduction rule to $G$. Then $\gamma_2(G)=\gamma_2(G')$.
\end{lemma}

\begin{proof}
Let $u, v\in V$ such that $N^3_G(u,v) \neq \emptyset$. Now if $N^2_G(u,v)=
\emptyset$ and $|N^3_G(u,v)|=1$, then $G'$ is same as $G$. So, without
loss of generality, assume that $\left| N^3_G(u,v) \right| > 1$ and
$N^2_G(u,v)\neq \emptyset$. Note that a vertex $x$ of $N^3_G(u,v)$ can be doubly
dominated by any two vertices from $N_G[x]\subseteq \{N^2_G(u,v)\cup
N^3_G(u,v) \cup \{u,v\}\}$. Again for any two vertices $x, y \in N^2_G(u,v)\cup
N^3_G(u,v) \cup \{u,v\}$, $N_G[x]\cup N_G[y] \subseteq N_G[u]\cup N_G[v]$. This
shows that we can double dominate each vertex of $N^3_G(u,v)$ in an optimal way
by selecting $u$ and $v$ only. This selection of $u$ and $v$ was forced by
the only vertex $w\in N^3_G(u,v)$ that remained in $G'$. We claim that $G$
contains a minimum double dominating set $D$ which does not contain any vertex
from $N^2_G(u,v)\cup N^3_G(u,v)$. First observe that there can not be three
or more vertices from $N^2_G(u,v)\cup N^3_G(u,v)$ in $D$. If it were, then we
could replace those three or more vertices by $u$ and $v$, thus contradicting
the minimality of $D$. Now for those two (or one) vertices from $N^2_G(u,v)\cup
N^3_G(u,v)$ in $D$, we can replace them by $u$ and (or) $v$.
Therefore, $G$ contains a minimum double dominating set $D$ which does not
contain any vertex from $N^2_G(u,v)\cup N^3_G(u,v)$. Clearly, this set
$D$ also forms a minimum double dominating set of $G'$. Hence,
$\gamma_2(G)=\gamma_2(G')$.
\end{proof}

In this reduction, for a pair of distinct vertices $u, v\in V$, we have
actually deleted at most $\min\{deg_G(u), deg_G(v)\}$ vertices. So, the
time taken is $\sum_{u,v\in V} \min\{deg_G(u), deg_G(v)\}$ for the whole reduction process. Since 
for a planar graph $\sum_{v\in V} deg_G(v)=O(n)$, where $n$ is the
number of vertices, we have the following lemma.

\begin{lemma}
For a planar graph having $n$ vertices, the reduction rule can be carried out in $O(n^3)$ time.
\end{lemma}

\subsubsection{A linear kernel}
\label{sssec:linearkernel}
In this subsection, we show that the reduction rule given in the previous
section yields a linear kernel for \textsc{Double Domination Problem} in planar
graphs. For this proof, first we find a ``maximal region decomposition'' of the
vertices $V'$ of the reduced graph $G'=(V',E')$ and then we show that
$|V'|=O(\gamma_2(G'))$. We start with some definitions regarding maximal region
decomposition following Alber et al.~\cite{Alberkernel}.

\begin{defi}\label{def_region}
Let $G=(V,E)$ be a plane graph. A closed subset of the plane is called a region
$R(u,v)$ between two vertices $u,v$ if the following properties are met:
\begin{enumerate}
\item the boundary of $R(u,v)$ is formed by two simple paths $P$ and $Q$
between $u$ and $v$ of length at most two edges, and
\item all the vertices which are strictly inside the region $R(u,v)$ are from
$N_G(u)\cap N_G(v)$.
\end{enumerate}
\end{defi}

The definition of a region is slightly different from the definition given in~\cite{Alberkernel}, where all the vertices which are strictly inside the region $R(u,v)$ are from $N_G(u)\cup N_G(v)$. Note that by the above definition, paths of length one or two between $u$ and $v$ can form a region $R(u,v)$. For a region $R=R(u,v)$, let $\partial(R)$ denote the boundary of $R$ and $V(R)$ denote the vertices inside or on the
boundary of $R$, i.e., $V(R)= \{u\in V |~ u~ \mbox{ is inside } R \mbox{ or  on } \partial(R)\}$.

\begin{defi}\label{def_regiondecom}
Let $G=(V,E)$ be a plane graph and $D\subseteq V$. A $D$-region decomposition
of $G$ is a set $\mathcal{R}$ of regions between pairs of vertices in $D$
such that
\begin{enumerate}

\item 
for $R(u,v)\in \mathcal{R}$ no vertices of $D$ (except $u$ and $v$) lies
in $V(R(u,v))$, and

\item 
for two regions $R_1, R_2\in \mathcal{R}$, $(R_1\cap
R_2)\subseteq (\partial(R_1)\cup \partial(R_2))$, i.e., they can intersect only
at the vertices on the boundary.

\end{enumerate}

For a $D$-region decomposition $\mathcal{R}$, we define
$V(\mathcal{R})= \cup_{R\in \mathcal{R}} V(R)$. A $D$-region decomposition
$\mathcal{R}$ is called maximal if there is no region $R$ such that
$\mathcal{R'}=\mathcal{R}\cup R$ is a $D$-region decomposition, where
$V(\mathcal{R})$ is a strict subset of $V(\mathcal{R'})$.
\end{defi}

First we observe an important property of a maximal $D$-region decomposition.

\begin{lemma}\label{lem V=V(R)}
Let $G=(V,E)$ be a plane graph with a double dominating set $D$ and let
$\mathcal{R}$ be a maximal $D$-region decomposition. Then $V=V(\mathcal{R})$.
\end{lemma}

\begin{proof}
Let $v\in V$ be a vertex such that $v\notin V(\mathcal{R})$. There can be two
cases -- $v \in D$ and $v \notin D$. First, let us assume that
$v\in D$. Since $D$ is a double dominating set of $G$, there exists another
vertex $x\in D$ such that $vx \in E$. Now, the path $P=(x,v)$ forms a region
$R$. Clearly $\mathcal{R}\cup R$ forms a $D$-region decomposition of $G$ which
contradicts the maximality of $\mathcal{R}$. Let us now consider the other
case $v\notin D$. Since $D$ is a double dominating set of $G$,
there exists $x, y\in D$ such that $vx, vy\in E$. In this case, the path
$P=(x,v,y)$ forms a region $R$. Here also,
$\mathcal{R}\cup R$ forms a $D$-region decomposition of $G$ which contradicts
the maximality of $\mathcal{R}$. Thus each vertex of $V$ is in $V(\mathcal{R})$,
that is, $V\subseteq V(\mathcal{R})$. Thus $V=V(\mathcal{R})$.
\end{proof}

It is obvious that, for a plane graph $G=(V,E)$ with a double dominating set
$D$, there exists a maximal $D$-region decomposition $\mathcal{R}$. Based on
Lemma \ref{lem V=V(R)}, we propose a greedy algorithm to compute a maximal
$D$-region decomposition, which is given in Algorithm \ref{maxregiondecomp}.
The algorithm basically ensures the properties of the region decomposition
mentioned in Definitions \ref{def_region} and \ref{def_regiondecom}.
\begin{algorithm}
\relsize{-1}{
\KwIn{A plane graph $G=(V,E)$ and a double dominating set $D\subseteq V$.}
\KwOut{A maximal $D$-region decomposition $\mathcal{R}$ of $G$.}
\Begin{
$V_{used}\leftarrow \emptyset$, $\mathcal{R}\leftarrow \emptyset$\;
\While {$V_{used}\neq V$}{
Select a vertex $x$ from $V\setminus V_{used}$\;
Consider the set $\mathcal{R}_x$ of all regions $S$ with the following properties:
\begin{enumerate}
\item $S$ is a region between $u$ and $v$, where $u,v \in D$.
\item $S$ contains $x$.
\item no vertex from $D\setminus \{u,v\}$ is in $V(S)$.
\item $(S\cup R)\subseteq (\partial(S)\cup \partial(R))$ for all $R\in \mathcal{R}$.
\end{enumerate}
Choose a region $S_x\in \mathcal{R}_x$ which is maximal in terms of vertices\;
$\mathcal{R}\leftarrow \mathcal{R}\cup \{S_x\}$\;
$V_{used}\leftarrow V_{used} \cup V(S_x)$\;
}
$\Return(\mathcal{R})$\;
}
\caption{REGION$\_$DECOMPOSITION$(G, D)$}  \label{maxregiondecomp} }
\end{algorithm}

Clearly Algorithm \ref{maxregiondecomp} output a maximal $D$-region decomposition in polynomial time.
Next, we show that for a given plane graph $G$ with a double dominating set $D$,
every maximal $D$-region decomposition contains at most $O(|D|)$ many
regions. For that purpose, we observe that a $D$-region decomposition induces a
graph in a very natural way.

\begin{defi}\label{def_indced graph}
The induced graph $G_{\mathcal{R}}=(V_{\mathcal{R}}, E_{\mathcal{R}})$ of a
$D$-region decomposition $\mathcal{R}$ of $G$ is the graph with possible
multiple edges which is defined as follows: $V_{\mathcal{R}}= D$ and $E_{\mathcal{R}}=\{(u,v)|$there is a region $R(u,v) \in \mathcal{R}$ between $u,v\in D \}$.
\end{defi}

Note that, since by Definition \ref{def_regiondecom} the regions of a $D$-region
decomposition do not intersect, the induced graph $G_{\mathcal{R}}$ of a
$D$-region decomposition $\mathcal{R}$ is a planar graph with multiple edges.
Next we bound the number of regions in a maximal $D$-region decomposition using the concept of \emph{thin planar graph} following Alber et al.~\cite{Alberkernel}.

\begin{defi}
A planar graph $G=(V,E)$ with multiple edges is thin if there exists a planar embedding such that if there are two edges $e_1, e_2$ between a pair of distinct vertices $v,w\in V$, then there must be two further vertices $u_1, u_2\in V$ which sit inside the two disjoint areas of the plane that are enclosed by $e_1$ and $e_2$.
\end{defi}

\begin{lemma}
Let $D$ be a double dominating set of a planar graph $G=(V,E)$. Then the induced graph 
$G_{\mathcal{R}}=(V_{\mathcal{R}}, E_{\mathcal{R}})$ of a maximal $D$-region decomposition 
$\mathcal{R}$ of $G$ is a thin planar graph.
\end{lemma}

\begin{proof}
Let $R_1$ and $R_2$ be two regions between two vertices $v,w\in D$ and
$e_1$ and $e_2$ be the corresponding multiple edges between two vertices $v, w\in V_{\mathcal{R}}$. Let $A$ be an area enclosed by $e_1$ and $e_2$. If $A$ contains a vertex $u\in D$, we are done. Suppose there is no vertex of $D$ in $A$. Now consider the following cases:
\begin{description}
\item[There is no vertex from $V\setminus D$ in $A$:] In this case, by combining the regions $R_1$ and $R_2$, we can form a bigger region which is a contradiction to the maximality of $\mathcal{R}$.

\item[There is a vertex $x\in (V\setminus D)$ in $A$:] In this case, if $x$ is double dominated by $v$ and $w$, then again we can combine the two regions $R_1$ and $R_2$ to get a bigger region. So, assume that $x$ is dominated by some vertex $u$ other than $v$ and $w$. Since $G$ is planar, $u$ must be in $A$ which contradicts the fact that $A$ does not contain any vertex from $D$.
\end{description}

Hence, combining both the cases we see that $G_{\mathcal{R}}$ is a thin planar graph.
\end{proof}

In~\cite{Alberkernel}, it is proved that for a thin planar graph $G=(V,E)$, we have $|E|\leq 3|V|-6$. Hence we have the following lemma.

\begin{lemma}\label{lem R less 3D}
For a plane graph $G$ with a double dominating set $D$, every maximal
$D$-region decomposition $\mathcal{R}$ contains at most $3|D|$ many regions.
\end{lemma}

Now, if we can bound the number of vertices that belongs to any region $R(u,v)$
of a maximal $D$-region decomposition $\mathcal{R}$ by some constant factor, we
are done. However, achieving this constant factor bound is not possible for any
plane graph $G$. But in a reduced plane graph, we can obtain this bound, as
shown in the following lemma.

\begin{lemma}\label{lem constnt vertex}
A region $R$ of a plane reduced graph contains at most $6$ vertices, that is,
$|V(R)|\leq 6$.
\end{lemma}
\begin{proof}
Let $R$ be the region between $u$ and $v$ and $\partial(R)=\{u, x, v, y\}$.
First note that $R$ contains at most two vertices from $N_{G'}^1(u,v)$ and the
only possibility of such vertices are $x$ and $y$. If there exists a vertex
$w \in N_{G'}^1(u,v)$, apart from $x$ and $y$, then $w$ has to have a neighbor
$z \notin N_{G'}(u,v)$. $z$ should be inside the region $R$ and hence, cannot be
double dominated. Now, because of the reduction rule, we can say that
$\left| N_{G'}^3 (u,v) \right| \le 1$. We consider the two cases:
\begin{description}
 \item[Case I ($\left| N_{G'}^3 (u,v) \right| = 1$):]
In this case, $\left| N_{G'}^2 (u,v) \right| = \emptyset$ by the reduction
rule. Hence, $|V(R)| \le 5$.
 \item[Case II ($\left| N_{G'}^3 (u,v) \right| = 0$):] In this case, we claim
that there can be at most two vertices from $N_{G'}^2 (u,v)$. If possible, let
$p,q,r \in N_{G'}^2 (u,v)$. Now all these three vertices must be adjacent
to either $x$ or $y$, which is not possible because of planarity. Hence,
in this case $|V(R)| \le 6$.
\end{description}
\end{proof}

\remove{
Otherwise, there exists a vertex
$z\in V(R)$ which is not doubly dominated. Now because of planarity, there
exists at most two vertices $p, q$ from $N_G^2(u,v)$ inside the region $R$. Also
note that if there exists a vertex $w$ from $N_G^3(u,v)$ inside $R$, then there
cannot be any vertex from $N_G^2(u,v)$ inside $R$ (because of the reduction
rule). Thus in the worst case, $V(R)$ contains at most $6$ vertices, that is,
$|V(R)|\leq 6$.}

First observe that, for a reduced graph $G'$ with a minimum double dominating
set $D$, by Lemma \ref{lem R less 3D}, there exists a maximal $D$-region
decomposition $\mathcal{R}$ with at most $3\cdot \gamma_2(G')$ regions. Also by
Lemma \ref{lem V=V(R)}, we have $V'=V(\mathcal{R})$ and by Lemma
\ref{lem constnt vertex}, we have for each region $|V(R)|\leq 6$. Thus we have
$|V'|=|V(\mathcal{R})|= |\cup_{R\in \mathcal{R}} V(R)| \leq \sum_{R\in
\mathcal{R}} \left|V(R) \right| \leq 6\cdot |\mathcal{R}| \leq 18\cdot
\gamma_{2}(G')$.
\remove{
\begin{eqnarray*}
|V|=|V(\mathcal{R})|= |\cup_{R\in \mathcal{R}} V(R)| \leq \sum_{R\in \mathcal{R}} |V(R)| \leq 6\cdot |\mathcal{R}|
 \leq 18\cdot \gamma_{2}(G).
\end{eqnarray*}}
Hence, we have the following theorem.
\begin{theo}\label{theodoublekernel}
For a reduced planar graph $G'=(V',E')$, we have $|V'|\leq 18\cdot
\gamma_{2}(G')$, that is, \textsc{Double Domination Problem} on planar graph
admits a linear kernel.
\end{theo}

\subsection{Liar's and $k$-tuple domination}
We first show that the number of vertices in a plane graph, $|V| =
O(\gamma_{LR}(G))$. In this respect, first we note that both the results in
Lemma \ref{lem V=V(R)} and Lemma \ref{lem R less 3D} are valid for any plane
graph $G$ and any double dominating set $D$. Since every liar's dominating set
is also a double dominating set, similar type of results hold for any plane
graph $G$ and any liar's dominating set $L$. \remove{The only part where we use
the reduced graph is to show that the number of vertices in a region $R$ of a
$D$-region decomposition is bounded above by a constant.} We claim that the
number of vertices in a region $R$ of a $L$-region decomposition is bounded
above by a constant. Let $R$ be a region between $u$ and $v$ and
$\partial(R) = \{u,x,v,y\}$. Note that in $V(R)$ there are two
vertices ($u$ and $v$) from $L$. Now, if there exists two vertices
$p, q \in V(R) \setminus \partial(R)$, then for the pair $p$ and $q$ condition
(ii) of liar's domination is violated. Hence, there is at most one vertex
in $V(R) \setminus \partial(R)$. Therefore, $|V(R)| \le 5$.
\remove{ But for liar's
dominating set $L$, we can prove this directly for any plane graph. Because each
region contains only two vertices of $L$ which, in turn, implies that
$|V(R)|\leq 5$. Otherwise there exists a pair of vertices $u,v\in V(R)$ such
that $|(N_G[u]\cup N_G[v])\cap D|\leq 2$.} Thus we have
$|V|=|V(\mathcal{R})|= |\cup_{R\in \mathcal{R}} V(R)| \leq \sum_{R\in
\mathcal{R}} |V(R)| \leq 5\cdot |\mathcal{R}|
 \leq 15\cdot |L|\leq 15\cdot \gamma_{LR}(G).$
\remove{
\begin{eqnarray*}
|V|=|V(\mathcal{R})|= |\cup_{R\in \mathcal{R}} V(R)| \leq \sum_{R\in \mathcal{R}} |V(R)| \leq 5\cdot |\mathcal{R}|
 \leq 15\cdot |L|\leq 15\cdot \gamma_{LR}(G).
\end{eqnarray*}
}
Hence, we have the following theorem
\begin{theo}\label{theoliarker}
For a planar graph $G=(V,E)$, $|V|\leq 15\cdot \gamma_{LR}(G)$.
\end{theo}
\remove{
In this subsection, we show similar result for \textsc{$k$-Tuple
Domination Problem} ($k\geq 3$) as we have done for \textsc{Liar's Domination
Problem} in the previous subsection. Since every $k$-tuple dominating set for
$k\geq 3$ is a liar's dominating set, it immediately follows from Theorem
\ref{theoliarker} that for a planar graph $G=(V,E)$, $|V|\leq 15\cdot
\gamma_{k}(G)$, where $k\geq 3$. But we can improve the constant by a little
margin as shown in the following theo.}
Since every $k$-tuple dominating set for $k\geq 3$ is a liar's dominating set,
we can use Theorem \ref{theoliarker}. But, we can improve the constant a little
bit.
\begin{theo}\label{theotupker}
For a planar graph $G=(V,E)$, $|V|\leq 12\cdot \gamma_{k}(G)$, where $k\geq 3$.
\end{theo}
\begin{proof}
Let $D$ be a minimum $k$-tuple dominating set of $G=(V,E)$. Since every
$k$-tuple dominating set is a double dominating set, by Lemma \ref{lem R less
3D} we can form a maximal $D$-region decomposition $\mathcal{R}$ of $G$
containing at most $3\cdot |D|$ many regions. Again by Lemma \ref{lem V=V(R)},
we have $V=V(\mathcal{R})$. Since each region contains only two vertices of $D$,
we have $|V(R)|\leq 4$. Otherwise there exists one vertex in $V(R)$ which is not
dominated by $k$ vertices of $D$. Hence $|V|\leq 4 \cdot |\mathcal{R}|\leq
12\cdot |D|\leq 12\cdot \gamma_{k}(G)$.
\end{proof}

\section{Linear kernels for bounded genus graphs}
In this section, we extend our results to bounded genus graphs to show that \textsc{$k$-Tuple Domination Problem} and \textsc{Liar's Domination Problem} admit a linear kernel. The notations in this section follow Section \ref{subsec:graphsurfacenotation}.

For double domination problem, we apply the same reduction rule on a graph $G$ with bounded genus $g$ to obtain the reduced graph $G'$. Note that the reduced graph $G'$ is also of bounded genus $g$. Let $G=(V,E)$ be an $n$-vertex $\Sigma$-embedded graph. It is easy to observe that, since $\sum_{v\in V}deg_G(v)=O(n+\textup{eg}(\Sigma))$, the reduced graph $G'=(V',E')$ can be computed in $O(n^3 + n^2\cdot \textup{eg}(\Sigma))$ time, where $|V|=n$. Next we show that $|V'|=O(\gamma_2(G')+g)$ which implies \textsc{Double Domination Problem} admits a linear kernel in bounded genus graphs.

To prove the above, we consider two cases. In the first case, we assume that the reduced $\Sigma$-embedded graph has representativity strictly greater than $4$. In the case when $rep(G)\leq 4$, we go by induction on the Euler genus of surface $\Sigma$. In the first case, the graphs are locally planar, i.e., all the contractable noose are of length less or equal to $4$. Since the boundary of the regions in planar case is less than or equal to $4$, the boundary $\partial(R)$ of any region $R$ of a $D$-region decomposition $\mathcal{R}$ is contractible. Hence the proof in the planar case can be extended in this case. Hence we have the following lemma.

\begin{lemma}\label{lemdoublerep4}
Let $G'=(V',E')$ be a reduced $\Sigma$-embedded graph where $\textup{rep}(G')>4$. Then $|V'|\leq 18(\gamma_2(G')+ \textup{eg}(\Sigma))$.
\end{lemma}
\begin{proof}
Let $D$ be a double dominating set of $G'$ and $\mathcal{R}$ is a maximal
$D$-region decomposition of $G'$. Forming a induced graph, $G_\mathcal{R}$ as in
case of double domination problem in planar graphs (Section
\ref{sssec:linearkernel}), we have $|\mathcal{R}|\leq 3\cdot(|D|+ eg(\Sigma))$.
Also, in this case, every vertex of $V'$ belongs to at least one region of
$\mathcal{R}$ and for a region $R$, $|V(R)| \le 6$.
Hence, we have $|V'|\leq 18(\gamma_2(G')+ \textup{eg}(\Sigma))$.
\end{proof}

Next consider the case where $3\leq \textup{rep}(G')\leq 4$. For a noose $N$ in
$\Sigma$, we define the graph $G_N=(V_N,E_N)$ as follows. First we consider the
graph $\mathcal{G}$ obtained from $G'=(V',E')$ by cutting along
$N$. Then for every $v\in N\cap V'$ if $v^i$, $i = 1, 2$, is not adjacent to a
pendant vertex, then we add a pendant vertex $u^i$ adjacent to $v^i$ to
form $G_N$. Clearly $G_N$ has genus less than that of $G'$. If we add all the
vertices of $V_N \setminus V'$ to a double dominating set $D$ of $G'$,
then we clearly obtain a double dominating set of $G_N$ and as,
$\textup{rep}(G')\leq 4$, $|V_N \setminus V'| \le 16$. Hence, $\gamma_2(G_N)\leq
\gamma_2(G')+|N\cap V'|\leq \gamma_2(G')+ 16$. Also note that if $G'$ is a
reduced graph, then so is $G_N$. Using these facts, we prove that \textsc{Double
Domination Problem} possesses a linear kernel when restricted to graphs with
bounded genus.

\begin{lemma}\label{lemdoublegenus}
For any reduced $\Sigma$-embedded graph $G'=(V',E')$ with
\textup{eg}$(\Sigma)\geq 1$, $|V'|\leq 18(\gamma_2(G')+ 32\cdot
\textup{eg}(\Sigma)-16)$.
\end{lemma}
\begin{proof}
We prove this result by induction on $\textup{eg}(\Sigma)$. Suppose
$\textup{eg}(\Sigma)=1$. If $\textup{rep}(G')>4$, then the result follows from
Lemma \ref{lemdoublerep4}. Otherwise Lemma \ref{lemcutting} implies that the
graph $G_N$, described above, is planar. Hence by Theorem
\ref{theodoublekernel}, we have $|V_N|\leq 18\cdot \gamma_2(G_N)$. Thus
$|V'|\leq |V_N|\leq 18(\gamma_2(G')+16)$.

Assume that $|V'|\leq 18(\gamma_2(G')+ 32\cdot \textup{eg}(\Sigma)-16)$ for any
$\Sigma$-embedded reduced graph $G'$ with eg$(\Sigma)\leq g-1$. Consider a
reduced $\Sigma$-embedded graph $G'$ with eg$(\Sigma)= g$. Now if rep$(G')>4$,
then again by Lemma \ref{lemdoublerep4}, we are done. Hence assume that
rep$(G')\leq 4$. By Lemma \ref{lemcutting}, either $G_N$ is the disjoint union
of graphs $G_1$ and $G_2$ that can be embedded in surfaces $\Sigma_1$ and
$\Sigma_2$ such that eg$(\Sigma)=$ eg$(\Sigma_1)+$ eg$(\Sigma_2)$ and
eg$(\Sigma_i) > 0$, $i = 1, 2$ (this is the case when $N$ is surface separating
curve), or $G_N$ can be embedded in a surface with Euler genus strictly smaller
than eg$(\Sigma)$ (this holds when N is not surface separating).

Let us consider the case where $G_N$ is the disjoint union of graphs $G_1=(V_1,E_1)$ and 
$G_2=(V_2,E_2)$ that can be embedded in surfaces $\Sigma_1$ and $\Sigma_2$. Since 
eg$(\Sigma_i)\leq g-1$ for $i=1,2$, we can apply the induction hypothesis on $G_i$. Thus we have,
\begin{align*}
|V'|\leq |V_N|&= |V_1|+ |V_2|& \\
&\leq \sum_{i=1}^{2} 18(\gamma_2(G_i)+ 32\cdot
\textup{eg}(\Sigma_i)-16)&\\
&\leq 18(\gamma_2(G_N)+ 32\cdot \textup{eg}(\Sigma)-32)& \text{as $G_1$ and $G_2$ are disjoint}\\
&\leq 18(\gamma_2(G')+ 32\cdot \textup{eg}(\Sigma)-16).&
\end{align*}

Next we consider the case where $G_N$ can be embedded in a surface $\Sigma'$
with Euler genus strictly smaller than $g$. In this case too, we can
apply induction hypothesis on $G_N$. Thus we have,
\begin{align*}
|V'|\leq |V_N| &\leq 18(\gamma_2(G_N)+ 32\cdot \textup{eg}(\Sigma')-16)&\\
&\leq 18(\gamma_2(G_N)+ 32\cdot (\textup{eg}(\Sigma)-1)- 16)&\\
&\leq 18(\gamma_2(G')+ 32\cdot \textup{eg}(\Sigma) -32)& \text{as $\gamma_2(G_N) \leq \gamma_2(G')+ 16$}\\
&\leq 18(\gamma_2(G')+ 32\cdot \textup{eg}(\Sigma)-16).&
\end{align*}

Thus we have proved that, $|V'|\leq 18(\gamma_2(G')+ 32\cdot
\textup{eg}(\Sigma)-16)$ for every $\Sigma$-embedded graph $G'=(V',E')$.
\end{proof}

Hence by Lemma \ref{lemdoublerep4} and Lemma \ref{lemdoublegenus}, we have the main result of this subsection.
\begin{theo}
\textsc{Double Domination Problem} admits a linear kernel for bounded genus graphs.
\end{theo}

For liar's domination problem, by Theorem \ref{theoliarker}, we have $|V|\leq 15\cdot \gamma_{LR}(G)$ in case of a planar graph $G=(V,E)$. Proceeding exactly in the same way as in the case of double domination, we can have the following theorem for a $\Sigma$-embedded graph $G$.

\begin{theo}\label{theoliargenus}
Let $G=(V,E)$ be a $\Sigma$-embedded graph. Then $|V|\leq 15(\gamma_{LR}(G)+ 32\cdot \textup{eg}(\Sigma))$.
\end{theo}

Since for any graph that admits a $k$-tuple dominating set ($k\geq 3$), $\gamma_{LR}(G)\leq \gamma_{k}(G)$, we have the following corollary of Theorem \ref{theoliargenus}.
\begin{coro}
For a $\Sigma$-embedded graph $G=(V,E)$, $|V|\leq 15(\gamma_{k}(G)+ 32\cdot \textup{eg}(\Sigma))$.
\end{coro}

\section{Conclusion}
\label{sec-conclusion}

In this paper, we first have proved that \textsc{$k$-Tuple Domination Problem}
and \textsc{Liar's Domination Problem} are \textsf{W[2]}-hard for general
graphs. Then we have shown that these two problems admit linear kernel for planar graphs 
and also for bounded genus graphs. It would be interesting to look for other graph classes 
where these problems admit efficient parameterized algorithms.

\section*{Acknowledgements}

The authors want to thank Venkatesh Raman and Saket Saurabh for some of their nice 
introductory expositions to parametrization.

\section{Appendix}
\label{sec-appendix}

\begin{lemma}\label{lem-liar-dom-planar-NP}
\textsc{Liar's Domination Problem} is NP-complete for planar graphs.
\end{lemma}
\begin{proof}
The reduction is from \textsc{Domination Problem} in planar graphs, which is known to be NP-complete\cite{GareyJohnson}. 
Let $G=(V,E)$ be a planar graph with $V=\{v_1, v_2, \ldots, v_n\}$ and $k$ be an integer. We construct an instance $G'=(V',E')$ 
and $k'$ of \textsc{Liar's Domination Problem} as follows: We add a set of $3n$ new vertices $S=\{x_i, y_i, z_i|1\leq i\leq n\}$ 
to the vertex set of $V$, i.e., $V'=V\cup S$ and the edge set of $G'$ is given by 
$E'=E\cup \{v_ix_i, x_iy_i, y_iz_i|1\leq i\leq n\}$. Note that, since $G$ is planar, so is $G'$. Also assume that $k'=k+3n$. 
In~\cite{Slater2009}, it is proved that $G$ has a dominating set of cardinality at most $k$ if and only if $G'$ has a liar's 
dominating set of cardinality at most $k'=k+3n$.

Thus, \textsc{Liar's Domination Problem} is NP-complete for planar graphs.
\end{proof}

\phantomsection
\bibliographystyle{alpha}
\addcontentsline{toc}{section}{Bibliography}
\bibliography{liardom}

\end{document}